\documentclass[12pt,draftcls,onecolumn]{IEEEtran}

\usepackage{amsmath,graphicx,color,psfrag,mathrsfs,bm}
\usepackage{hyperref}
\usepackage{pstool}
\usepackage{longtable}          
\usepackage{amsthm}
\usepackage{amssymb}

\usepackage{multicol}
\newtheorem{thm}{Theorem}[section]
\newtheorem{kor}[thm]{Corollary}
\newtheorem{lem}[thm]{Lemma}
\newtheorem{prop}[thm]{Proposition}

\newtheorem{exm}[thm]{Example}

\theoremstyle{definition}

\theoremstyle{problem}
\newtheorem{prbm}{Problem}[section]
\theoremstyle{remark}
\newtheorem{rem}{Remark}[section]

\begin{document}
%
\title{Minimax Robust Decentralized Detection in Parallel Sensor Networks}

\author{G\"{o}khan G\"{u}l,~\IEEEmembership{Member,~IEEE}\\
         Technische Universit\"{a}t Darmstadt, 64283, Darmstadt, Germany\\ E-mail: ggul@spg.tu-darmstadt.de}

\maketitle

\begin{abstract}
Minimax robust decentralized detection is studied for parallel sensor networks. Random variables corresponding to sensor observations are assumed to follow a distribution function, which belongs to an uncertainty class. It has been proven that, for some uncertainty classes, if all probability distributions are absolutely continuous with respect to a common measure, the joint stochastic boundedness property, which is the fundamental rule for the derivations in Veerevalli's work, does not hold. This raises a natural question whether minimax robust decentralized detection is possible if the uncertainty classes do not own this property. The answer to this question has been shown to be positive, which leads to a generalization of the work of Veerevalli. Moreover, due to a direct consequence of Tsitsiklis's work, quantization functions at the sensors are not required to be monotone. For the proposed model, some specific examples have been provided and possible generalizations have been discussed.

\end{abstract}


%
\IEEEpeerreviewmaketitle


\section{Introduction}
\label{sec:intro}
In simple binary hypothesis testing the design of optimum decision rules requires the exact knowledge of the conditional probability distributions under each hypothesis. However, in practice, complete knowledge of the observation statistics is often not available, such as occurs with the presence of outliers or due to model mismatch. In these cases, a reasonable approach is to represent each hypothesis by a set or class of distributions and determine the optimum decision rule via minimizing the worst case performance. Such tests are called minimax robust tests and they often have the property of guaranteeing a certain level of detection performance irrespective of the actual state of the observation statistics. Because of this property, minimax robust tests are often essential for the design of systems that have to function reliably in harsh environments or in environments which cannot be modeled accurately \cite{Levy}.\\
The first and probably the most fundamental work in robust hypothesis testing was developed by Huber in 1965 \cite{hube65}. He showed that the minimax robust test for the $\epsilon$-contaminated classes of distributions and the uncertainty classes with respect to the total variation distance were clipped likelihood ratio tests, where the likelihood ratio was obtained between so called least favorable distributions from the respective uncertainty classes. In his follow-up work, Huber extended the class of distributions to five, from where the same conclusions could be made \cite{hube68}. The most general classes of distributions for which the clipped likelihood ratio is the minimax robust test are the two alternating capacities, published by Huber and Strassen \cite{hube73}.\\
In addition to robustness, another important aspect is to include multiple decision makers (physical sensors) into the decision making process. In many practical applications, such as radar, wireless communication, or seismology, more than a single sensor is available and it is well known that if the events of interest are independent, the system error probability decreases exponentially with the number of sensors \cite{Chernoff}. Although the benefits of robust distributed detection are obvious, progress made since 1980's has been insignificant \cite{Varshney}. The earliest study in this field was conducted by Geroniotis, who considered a distributed detection network without a fusion center (DDN-WoF) for a fixed sample size and a sequential discrete time robust detection for two sensors \cite{Ger2}. In \cite{Ger3}, Geraniotis and Chau studied the robustness of distributed detection network with a fusion center (DDN-WF) and sequential data fusion where the emphasis was on the selection of robust fusion rules. In their recent work \cite{Ger4}, Geraniotis and Chau generalized most of their results presented in \cite{Ger3}.\\
All Huber's classes of distributions satisfy joint stochastic boundedness property. Based on this observation it was proven in \cite{Poor1} that for jointly stochastically bounded classes of distributions, there exist least favorable distributions for DDN-WF if the individual sensors employ robust tests. Moreover, the authors formalized necessary conditions that need to be satisfied by the cost assignment procedure for DDN-WoF. The results derived in \cite{Poor1} generalize the DDN-WoF-results of Geroniotis \cite{Ger2} to a network of more than two sensors and to more general cost functions. Furthermore, the results of \cite{Poor1} also generalizes the DDN-WF-results of Geraniotis and Chau \cite{Ger3, Ger4} to non-Bayesian formulation, non-binary decisions, non-identical sensor decisions and non-asymptotic case, both in terms of the number of sensors as well as the number of observations. Recent studies in robust decentralized detection consider different network topologies, for instance tandem sensor networks, where asymptotic analysis is of great interest \cite{Tsi3,serial} and application of earlier results to scenarios with constraints such as power \cite{park}, communication rate \cite{gul2}, or local optimality \cite{censoring}.\\
In this paper, a more comprehensive solution to minimax robust decentralized detection problem is provided. The network topology is parallel with a finite number of sensors and a fusion center. Each sensor in the sensor network collects a finite number of samples characterizing either the null or the alternative hypothesis and gives a decision which is possibly multi-level. The proposed scheme includes the work of Veerevalli et. al. \cite{Poor1} as a special case since in our work the two conditions: $1)$ joint stochastic boundedness property and $2)$ monotone sensor quantization functions are not necessarily required. Moreover, generalizations to Neyman-Pearson detection, repeated observations and different network topologies are also discussed.\\
The organization of this paper is as follows. In Section~\ref{motiv_prob}, the motivation and the problem definition is given. In Section~\ref{minimax} the theory behind the solution of minimax robust decentralized detection problem is introduced. In Section~\ref{examples} specific examples are given. In Section~\ref{general} possible generalizations of the theory is discussed, and finally in Section~\ref{conclusion} the paper is concluded.

\section{Motivation and Problem Definition}\label{motiv_prob}
Binary minimax decentralized detection is studied for parallel sensor networks as illustrated in Figure~\ref{fig1}. The hypotheses $\mathcal{H}_0$ and $\mathcal{H}_1$ are associated with the probability measures $P_0$ and $P_1$, which have the density functions $p_0$ and $p_1$, respectively. Here, and in the following sections every probability measure e.g. $P[\cdot]$ will be associated with its distribution function $P(\cdot)$ i.e. $P(y)=P[Y\leq y]$ for the random variable $Y$ and the observation $y$. The detailed structure of the sensor network will be presented after stating the motivation. The following remark, and lemmas will be used in the rest of the paper.

\begin{rem}\label{rem1}
Let $X$ and $Y$ be two random variables defined on the same measurable space $(\Omega,{\mathcal{A}})$, having cumulative distribution functions $P_X$ and $P_Y$, respectively. $X$ is called stochastically larger than $Y$, i.e. $X\succeq Y$, if $P_Y(x)\geq P_X(x)$ for all $x$.
\end{rem}

\begin{lem}\label{cor1}
For every non-decreasing function $\upsilon$, $X\succeq Y\Longleftrightarrow\upsilon(X)\succeq \upsilon(Y)$, hence $X\succeq Y \Longleftrightarrow \mathrm{E}[\upsilon(X)]\geq \mathrm{E}[\upsilon(Y)]$.
\end{lem}
Proof of Lemma \ref{cor1} is simple and can be found for example in \cite[pp. 4-5]{Wolfstetter96}.

\begin{lem}\label{lem1}
Let $X_1$, $X_2$, $Y_1$ and $Y_2$ be four random variables on $(\Omega,{\mathcal{A}})$, out of which $X_1$ and $X_2$, and $Y_1$ and $Y_2$ are independent. If $X_1\succeq Y_1$ and $X_2\succeq Y_2$, then $X_1+X_2\succeq Y_1+Y_2$.
\end{lem}

\begin{proof}
From Remark~\ref{rem1}, we have $P_{Y_1}(x)\ge P_{X_1}(x)$ and $P_{Y_2}(x)\ge P_{X_2}(x)$ for all $x$. Hence,
\begin{align}
P_{Y_1+Y_2}(z)&=\int_{-\infty}^{+\infty}P_{Y_1}(z-x)dP_{Y_2}(x)
\ge \int_{-\infty}^{+\infty}P_{X_1}(z-x)dP_{Y_2}(x)\nonumber\\
&=\iint_{x+y\le z}dP_{X_1}(x)dP_{Y_2}(y)
= \int_{-\infty}^{+\infty}P_{Y_2}(z-y)dP_{X_1}(y)\nonumber\\
&\ge \int_{-\infty}^{+\infty}P_{X_2}(z-y)dP_{X_1}(y)=P_{X_1+X_2}(z).
\end{align}
\end{proof}

\subsection{Motivation}
It is stated by Huber \cite{hube73} that if the classes of distributions are constructed such that every distribution in the uncertainty class is absolutely continuous with respect to a dominating measure and the domain of the uncertainty classes are uncountably infinite, the stochastic boundedness property may fail. This property specifies minimax robustness in all Huber's papers \cite{hube65,hube68,hube73} and it is a precondition for the design of minimax robust decentralized detection in \cite{Poor1}. This leads to the following questions:
\begin{enumerate}
  \item Are there classes of distributions for which joint stochastic boundedness property fails?
  \item Is minimax robust decentralized detection possible in this case?
\end{enumerate}
In the sequel, two examples of uncertainty classes are provided, where the stochastic boundedness property holds and fails, respectively. For both examples, every distribution in the uncertainty classes is absolutely continuous with respect to the related nominal measure. The second question will be addressed starting from the next section.
\begin{exm}
Let $P_0=\mathcal{N}(t_l^0,\sigma^2)$ and $P_1=\mathcal{N}(t_u^1,\sigma^2)$ be the nominal distributions and
\begin{align}\label{eq1}
\mathscr{P}_0&=\{Q_0:Q_0=\mathcal{N}(\mu_0,\sigma^2),\,\mu_0\in[t_l^0,t_u^0]\},\nonumber\\
\mathscr{P}_1&=\{Q_1:Q_1=\mathcal{N}(\mu_1,\sigma^2),\,\mu_1\in[t_l^1,t_u^1]\}.
\end{align}
Then, $\hat {Q}_0=\mathcal{N}(t_u^0,\sigma^2)$ and $\hat {Q}_1=\mathcal{N}(t_l^1,\sigma^2)$ are the least favorable distributions satisfying the joint stochastic boundedness property
\begin{align}\label{eq2}
&Q_0[\hat{l}(Y)\leq t]\geq \hat{Q}_0[\hat{l}(Y)\leq t],\,\,\, \forall t\in\mathbb{R}_{\geq 0}, \forall Q_0\in\mathscr{P}_0\\
&Q_1[\hat{l}(Y)\leq t]\leq \hat{Q}_1[\hat{l}(Y)\leq t],\,\,\, \forall t\in\mathbb{R}_{\geq 0}, \forall Q_1\in\mathscr{P}_1.\label{eq3}
\end{align}
where $\hat l=d \hat {Q}_1/d \hat {Q}_0$ is the robust likelihood ratio function.
\end{exm}
\begin{proof}
Since for $$g(\mu_0,t)=\frac{\mu_0}{\sqrt{2}\sigma}-\frac{1}{2\sqrt 2 \sigma}\left(t_u^0+t_l^1+\frac{2\sigma^2 t}{t_l^1-t_u^0}\right)$$
we have
\begin{equation}\label{eq4}
f(\mu_0,t)=Q_0[\log\hat{l}(Y)\leq t]=\frac{1}{\sqrt \pi}\int^{\infty}_{g(\mu_0,t)}e^{-x^2}\mathrm d x
\end{equation}
and
\begin{equation}\label{eq5}
f(\mu_0+\epsilon,t)-f(\mu_0,t)=-\frac{1}{\sqrt \pi}\int^{g(\mu_0+\epsilon,t)}_{g(\mu_0,t)}e^{-x^2}\mathrm d x
\end{equation}
is decreasing in $\epsilon$ for every $t$, \eqref{eq2} holds. The proof for \eqref{eq3} is similar and is omitted.
\end{proof}
\begin{exm}
The second example will be stated with the following proposition.
\begin{prop}\label{prop1}
Let the uncertainty classes be
\begin{equation}\label{eq6}
{\mathscr{P}}_j=\{Q_j:D(Q_j,P_j)\leq \epsilon_j\},\quad j\in\{0,1\},
\end{equation}
where
\begin{equation*}
D(Q_j,P_j)=\int_{\Omega}\ln(d Q_j/d P_j)\mathrm{d} Q_j,\quad j\in\{0,1\}
\end{equation*}
is the KL-divergence. Then, there exists no pair of LFDs ($\hat {Q}_0,\hat {Q}_1$) such that \eqref{eq2} and \eqref{eq3} hold.
\end{prop}

\end{exm}
\begin{proof}
The claim can be proven by contradiction. Assume that there exists such a pair of LFDs. Then, the same pair must satisfy
\begin{align}\label{eq7}
\hat{Q}_0&=\arg \max_{Q_0\in\mathscr{P}_0} \mathrm{E}_{Q_0}\ln(d\hat{Q}_1/d\hat{Q}_0),\nonumber\\
\hat{Q}_1&=\arg \min_{Q_1\in\mathscr{P}_1} \mathrm{E}_{Q_1}\ln(d\hat{Q}_1/d\hat{Q}_0)
\end{align}
by applying Remark~\ref{rem1} and Lemma~\ref{cor1} in \eqref{eq2} and \eqref{eq3}. By Huber and Strassen \cite[Theorem 7.1]{hube73}, see also \cite{gulbook}, \eqref{eq7} is equivalent to
\begin{align}\label{eq8}
\hat{Q}_0&=\arg \max_{Q_0\in\mathscr{P}_0} \mathrm{E}_{Q_0}\ln(d\hat{Q}_1/d Q_0),\nonumber\\
\hat{Q}_1&=\arg \min_{Q_1\in\mathscr{P}_1} \mathrm{E}_{Q_1}\ln(d Q_1/d\hat{Q}_0).
\end{align}
By Dabak, \cite{dabak}, see also \cite{Levy}, the pair of distributions solving \eqref{eq8} are given by
\begin{equation}\label{eq9}
\hat {q}_0=\frac{{p_0}^{1-u} {p_1}^u}{\int_{\Omega}{p_0}^{1-u} {p_1}^u},\quad \hat {q}_1=\frac{{p_0}^v{p_1}^{1-v} }{\int_{\Omega}{{p_0}^v p_1}^{1-v} }
\end{equation}
with respect to their density functions, where $u$ and $v$ are parameters to be determined such that
\begin{equation}\label{eq10}
D(\hat{Q}_0,P_0)=\epsilon_0,\quad D(\hat{Q}_1,P_1)=\epsilon_1.
\end{equation}
However, the test based on $\hat l=\hat {q}_1/\hat {q}_0$ is still a nominal likelihood ratio test \cite{dabak, Levy}, though with a modified threshold, and therefore it is not minimax robust \cite{hube81}. Hence, no pair of distributions is jointly stochastically bounded for the KL-divergence neighborhood.
\end{proof}
Notice that a minimax robust test for the KL-divergence exits and the corresponding test is unique and is randomized \cite{gul5}. Since this test is not equivalent to any deterministic likelihood ratio test, it also does not satisfy \eqref{eq2} and \eqref{eq3}.

\subsection{Problem Definition}
\begin{figure}[ttt]
  \centering
  \centerline{\includegraphics[width=70mm]{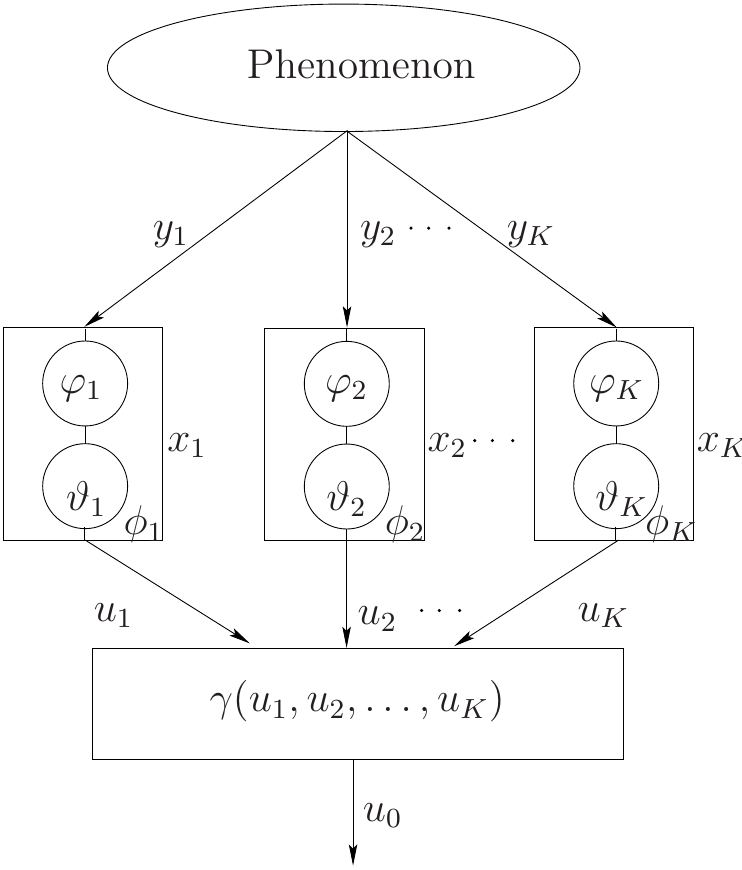}}
\caption{Distributed detection network with $K$ decision makers, each represented by the decision rule $\phi_i$, which is the composition of two functions $\varphi_i$ and $\vartheta_i$, and a fusion center associated with the fusion rule $\gamma$.\label{fig1}}
\end{figure}
Consider a decentralized detection network with a parallel topology as shown in Figure~\ref{fig1}. There are $K$ decision makers observing a certain phenomenon, and a fusion center. All random variables $Y_i$ corresponding to the observations $y_i$ take values on a measurable space $(\Omega_i,{\mathcal{A}}_i)$ and are assumed to be independent under each hypothesis, but not necessarily identical. Every decision maker $\phi_i$ is assumed to be composed of two possibly random functions $\varphi_i:\Omega_i\rightarrow A_i$, where $A_i\in\mathcal{A}_i$ and $\vartheta_i:A_i\rightarrow S_i \subset\mathbb{N}$. Given an observation $y_i$, every sensor transmits its own decision $u_i=\phi_i(y_i)=\vartheta_i(\varphi_i(y_i))$ to the fusion center. The fusion center, i.e. $\gamma$ then makes the final binary decision $u_0$ based on all decisions $u_1,\ldots,u_K$ that are received. The technical details related to the random variables $Y_i$, $X_i$ and $U_i$ corresponding to the observations $y_i$, $x_i$ and $u_i$, respectively, which are shown in Figure~\ref{fig1}, are detailed below:

\begin{itemize}
\item Under each hypothesis $\mathcal{H}_j$, the random variables $Y_i^j$, $X_i^j=\varphi_i(Y_i^j)$ and $U_i^j=\vartheta_i(\varphi(Y_i^j))$ follow the distributions $Q_{j}^{Y_i^j}$, $Q_{j}^{X_i^j}$ and $Q_{j}^{U_i^j}$ which belong to the uncertainty classes, $\mathscr{P}_j^{Y_i^j}$, $\mathscr{P}_j^{X_i^j}$ and $\mathscr{P}_j^{U_i^j}$, respectively. In order to avoid cumbersome notation, the distributions will be denoted by $Q_0^i$ and $Q_1^i$, and the uncertainty classes by $\mathscr{P}_0^{i}$ and $\mathscr{P}_1^{i}$ omitting the random variables in superscripts.

\item Similarly, the distributions $\mathbf{Q}_{0}=(Q_0^1,\ldots,Q_0^K)$ and $\mathbf{Q}_{1}=(Q_1^1,\ldots,Q_1^K)$ belong to the product uncertainty classes $\mathscr{P}_0=\mathscr{P}_0^1\times\ldots\times\mathscr{P}_0^K$ and $\mathscr{P}_1=\mathscr{P}_1^1\times\ldots\times\mathscr{P}_1^K$, respectively.

\item $\mathbf{Y}^j=(Y^j_1,\ldots,Y^j_K)$, $\mathbf{X}^j=(X^j_1,\ldots,X^j_K)$ and $\mathbf{U}^j=(U^j_1,\ldots,U^j_K)$ are the multivariate random variables under the hypothesis $\mathcal{H}_j$, and $\mathbf{Y}$, $\mathbf{X}$ and $\mathbf{U}$ are defined similarly without the index $j$. The vector notation is also applied to the collection of decision rules $\boldsymbol{\phi}=(\phi_1,\ldots,\phi_K)$ where $\boldsymbol{\phi}\in\boldsymbol{\Delta}=\Delta_1\times\ldots\times\Delta_K$.

\item The stochastically larger sign $\succeq$ is extended to vector notation $\boldsymbol{\succeq}$, e.g., $$\mathbf{U}^j\boldsymbol{\succeq}\mathbf{U}^j\Longrightarrow {U^j_i}\succeq{U^j_i},\quad \forall i.$$ and $\hat{(\cdot)}$ indicates the LFDs, e.g. $\hat {U}^j_i$ is the random variable $U_i$ which follows $\hat {Q}_j$.
\end{itemize}
Moreover, the nominal and the robust likelihood ratio functions for each decision maker $i$ are denoted by $l_i$ and $\hat{l}_i$, respectively. Let the false alarm and miss detection probabilities be defined as $P_F$ and $P_M$. Then, the minimum error probability can be written as
\begin{align}\label{eq11}
P_E(\boldsymbol\phi,\gamma,\mathbf{Q}_0,\mathbf{Q}_1)=&P(\mathcal{H}_0)P_F(\boldsymbol\phi,\gamma,\mathbf{Q}_0)+P(\mathcal{H}_1)P_M(\boldsymbol\phi,\gamma,\mathbf{Q}_1)\nonumber\\
=&P(\mathcal{H}_0)\int \gamma(\boldsymbol\phi(Y))\mathrm{d}\mathbf{Q}_0+(1-P(\mathcal{H}_0))\int (1- \gamma(\boldsymbol\phi(Y)))\mathrm{d}\mathbf{Q}_1
\end{align}
where $t=P(\mathcal{H}_0)/P(\mathcal{H}_1)$ is the threshold. Accordingly, a solution to the following problem is seeked:
\begin{prbm}
The minimax optimisation problem is stated as follows:
\begin{equation}\label{eq12}
\inf_{\boldsymbol\phi\in\boldsymbol\Delta,\gamma}\sup_{(\mathbf{Q}_0, \mathbf{Q}_1)\in {\mathscr{P}}_0\times{\mathscr{P}}_1}P_E(\boldsymbol\phi,\gamma,\mathbf{Q}_0,\mathbf{Q}_1).
\end{equation}
\end{prbm}
\noindent A solution to this problem results in the saddle value inequalities (see \cite{Poor1} for details):
\begin{equation}\label{eq13}
 P_{E}(\hat{\mathbf{Q}}_{0},\hat{\mathbf{Q}}_{1},\boldsymbol{\phi},\gamma)\geq P_{E}(\hat{\mathbf{Q}}_{0},\hat{\mathbf{Q}}_{1},\hat{\boldsymbol{\phi}},\hat{\gamma})\geq P_{E}(\mathbf{Q}_0,\mathbf{Q}_1,\hat{\boldsymbol{\phi}},\hat{\gamma}).
\end{equation}
The left and the right inequalities indicate the minimisation and the maximisation defined in \eqref{eq12}, respectively. The right inequality in \eqref{eq13} also implies
\begin{align}\label{eq14}
&P_{F}(\hat{\mathbf{Q}}_{0},\hat{\boldsymbol{\phi}},\hat{\gamma})\geq P_{F}(\mathbf{Q}_0,\hat{\boldsymbol{\phi}},\hat{\gamma}),\quad \forall  \mathbf{Q}_0\in\mathscr{P}_0\nonumber\\
&P_{M}(\hat{\mathbf{Q}}_{1},\hat{\boldsymbol{\phi}},\hat{\gamma})\geq P_{M}(\mathbf{Q}_1,\hat{\boldsymbol{\phi}},\hat{\gamma}),\quad \forall  \mathbf{Q}_1\in\mathscr{P}_1
\end{align}
since $P_{E}$ is distinct in $\mathbf{Q}_{0}$ and $\mathbf{Q}_{1}$. The converse is also true hence, \eqref{eq14} $\Longleftrightarrow$ \eqref{eq13}, if $\hat{\boldsymbol{\phi}}$ and $\hat\gamma$ jointly minimise $P_E$. The following section details the conditions that need to be satisfied by $\varphi_i$, $\vartheta_i$ and $\gamma$ such that \eqref{eq13} holds, see Figure~\ref{fig1}.

\section{Minimax Robust Decentralized Detection}\label{minimax}
Error minimising decision rules $\hat{\boldsymbol{\phi}}$ and the fusion rule $\hat\gamma$ are known to be the likelihood ratio tests. The conditions that need to be satisfied for \eqref{eq14} to hold are twofold:
 \begin{enumerate}
   \item  Conditions defined on $\mathbf{U}$ and from $\mathbf{U}$ to $U_0$ via the fusion rule $\gamma$.
   \item  Conditions defined from $\mathbf{Y}$ to $\mathbf{U}$ via $\varphi_i$ and $\vartheta_i$ such that the conditions defined in 1) hold.
 \end{enumerate}
The following theorem details 1), whereas the next two theorems suggest two possible solutions for 2).

\begin{thm}\label{theorem1}
The inequalities defined by \eqref{eq14} hold if $\hat{\boldsymbol{\phi}}$ results in
\begin{enumerate}
\item $\hat{\mathbf{U}}^0\boldsymbol{\succeq}\mathbf{U}^0$ and $\mathbf{U}^1\boldsymbol{\succeq}\hat{\mathbf{U}}^1$,
\item $\hat{l}_i=\hat{q}_1^i/\hat{q}_0^i$ is almost everywhere $\mu=\hat{Q}_1^i+\hat{Q}_0^i$ equal to a monotone non-decreasing function for every $i$.
\end{enumerate}
\end{thm}

\begin{IEEEproof}\label{proofoftheorem1}
Since $U_1,\ldots,U_K$ are all mutually independent random variables, the optimum fusion rule $\hat\gamma$ at the fusion center is to make a decision based on
\begin{equation*}
\hat{l}(\mathbf{U})=\frac{\hat q_1(\mathbf{U})}{\hat q_0(\mathbf{U})}=\prod_{i=1}^K\frac{\hat q_1^i({U}_i)}{\hat q_0^i({U}_i)}\stackrel{\mathcal{H}_1}{\underset{\mathcal{H}_0}{\gtrless}}t
\end{equation*}
which is equivalent to
\begin{equation}\label{eq145}
\log\hat{l}(\mathbf{U})=\log\prod_{i=1}^K\frac{\hat q_1^i({U}_i)}{\hat q_0^i({U}_i)}=\sum_{i=1}^K \log\hat{l}_i(U_i)\stackrel{\mathcal{H}_1}{\underset{\mathcal{H}_0}{\gtrless}}t.
\end{equation}
From condition 2), recall that $\hat{l}_i$ is monotone non-decreasing, $\log \hat{l}_i$ is also monotone non-decreasing for all $i$. Using Lemma~\ref{cor1} in condition 1) with $\upsilon=\log \hat{l}_i$, all summands in \eqref{eq145} satisfy
\begin{align}\label{eq146}
\log \hat{l}_i(\hat{U}_{i}^0)\succeq\log \hat{l}_i\left(U_i^0\right),\quad \forall i, Q_0^i\in\mathscr{P}_0^i,\nonumber\\
\log \hat{l}_i\left(U_i^1\right)\succeq \log \hat{l}_i(\hat{U}_i^1),\quad \forall i, Q_1^i\in\mathscr{P}_1^i.
\end{align}
Accordingly, by applying Lemma~\ref{lem1} to both inequalities in \eqref{eq146} inductively, i.e. to the pairs of random variables iteratively, leads to
\begin{align}\label{eq15}
\sum_{i=1}^K\log \hat{l}_i(\hat{U}_i^0)\succeq \sum_{i=1}^K\log \hat{l}_i\left(U_i^0\right),\quad \forall Q_0^i\in\mathscr{P}_0^i,\nonumber\\
\sum_{i=1}^K\log \hat{l}_i\left(U_i^1\right)\succeq \sum_{i=1}^K \log \hat{l}_i(\hat{U}_i^1),\quad \forall Q_1^i\in\mathscr{P}_1^i.
\end{align}
Let $\hat{Q}_j$ and $Q_j$ be the probability distributions of the random variable $\sum_{i=1}^K\log \hat{l}_i\left(U_{i}\right)$, when $U_i$ is distributed as $\hat{Q}_j^i$ and $Q_j^i$, respectively. Then, the stochastic ordering stated by \eqref{eq15}, cf. Remark~\ref{rem1}, leads to
\begin{align}\label{eq16}
\hat{Q}_0\left[\sum_{i=1}^K\log \hat{l}_i\left(U_i\right)>t\right]\geq {Q}_0\left[\sum_{i=1}^K\log \hat{l}_i\left(U_i\right)>t\right],\quad \forall t, Q_0,\nonumber\\
\hat{Q}_1\left[\sum_{i=1}^K\log \hat{l}_i\left(U_i\right)\leq t\right]\geq Q_1\left[\sum_{i=1}^K\log \hat{l}_i\left(U_{i}\right)\leq t\right] ,\quad \forall t, Q_1.
\end{align}
The inequalities in \eqref{eq16} imply the assertion, hence, the proof is complete.
\end{IEEEproof}

The sufficient conditions amongst the random variables $U_1,\ldots,U_K$ as well as from $\mathbf U$ to $U_0$ have been established with Theorem~\ref{theorem1}. Next, the sufficient conditions from $\mathbf Y$ to $\mathbf U$ will be stated with a suitable choice of the decision rules $\hat{\boldsymbol{\phi}}$, i.e. with $\varphi_i$ and $\vartheta_i$.

\begin{thm}\label{theorem2}
If the function $\varphi_i: \Omega_i\rightarrow A_i$ with the mapping $Y_i\mapsto\hat{l}(Y_i)$ results in
\begin{equation}\label{eq17}
\hat{X}_i^0\succeq X_i^0\quad\mbox{and}\quad X_i^1\succeq \hat{X}_i^1,\quad\forall i, Q_j^i\in\mathscr{P}_j^i
\end{equation}
and if $\vartheta_i$ is a monotone non-decreasing function,
\begin{equation}\label{eq18}
U_i=\vartheta_i(X_i)= \begin{cases} 0, & X_i<t_0^i \\ d, & t_{d-1}^i\leq X_i<t_{d}^i\\ D_{i} &  X_i>t_{D_{i}-1} \end{cases},\quad \forall d\in S_i=\{1,\ldots,D_i-1\}
\end{equation}
then, the two conditions described in Theorem~\ref{theorem1} hold and therefore all conclusions therein follow.
\end{thm}

\begin{proof}\label{proofoftheorem2}
The mapping $\vartheta_i$ is monotone non-decreasing and from Lemma~\ref{cor1}, it follows that
\begin{equation*}
\hat{U}_i^0\succeq U_i^0\quad\mbox{and}\quad U_i^1\succeq \hat{U}_i^1,\quad \forall i, Q_j^i\in\mathscr{P}_j^i.
\end{equation*}
The function $\hat{l}_{i}=\hat{q}_1^i/\hat{q}_0^i$ is a.e. equal to a monotone non-decreasing function for all $i$ as
\begin{align*}
\hat{l}_{i}(U_{i}=d)&=\frac{\hat{Q}_1^i[t_{d-1}\leq X_i<t_d]}{\hat{Q}_0^i[t_{d-1}\leq X_i<t_d]}\leq \frac{\hat{Q}_1^i[t_d\leq X_i<t_{d+1}]}{\hat{Q}_0^i[t_d\leq X_i<t_{d+1}]}\nonumber\\
&=\hat{l}_i(U_{i}=d+1)
\end{align*}
holds for all $d$, since
\begin{align*}
\frac{\hat{Q}_1^i[t_{d-1}\leq X_i<t_d]}{\hat{Q}_0^i[t_{d-1}\leq X_i<t_d]}&=\frac{1}{\hat{Q}_0^i[t_{d-1}\leq X_i<t_d]}\int_{\{t_{d-1}\leq X_i<t_d\}}\mathrm{d}\hat{Q}_1^i\nonumber\\
&=\frac{1}{\hat{Q}_0^i[t_{d-1}\leq X_i<t_d]}\int_{\{t_{d-1}\leq X_i<t_d\}}X_i\mathrm{d}\hat{Q}_0^i\nonumber\\
&=\mathrm{E}_{\hat{Q}_0^i}[X_i|t_{d-1}\leq X_i <t_d]
\end{align*}
is a number between $t_{d-1}$ and $t_{d}$. Obviously, the result also applies to the end points, i.e. $\hat{l}_{i}(U_i=0)$ and $\hat{l}_{i}(U_i=D_i)$, considering the intervals $(0,t_0^i)$ and $(t_{D_{i}-1},\infty)$, respectively.
\end{proof}
\noindent The results of Theorem~\ref{theorem2} can be extended to include non-monotone $\vartheta_i$ in case a well defined permutation function is applied at the fusion center. This is stated with the following corollary.

\begin{kor}\label{corollary}
Let $\vartheta_i$ be any bijective mapping from the set of non-overlapping intervals of $A_i$ to the set $S_i$. Then, there exists a permutation mapping $\varrho_i$ at the fusion center such that the two conditions described in Theorem~\ref{theorem1} hold and all conclusions therein follow.
\end{kor}

\begin{proof}\label{proofofcorollary}
Since $\vartheta_i$ is a bijective mapping, the total number of intervals of $A_i$ must have the same cardinality with the cardinality of $S_i$. Then, for every decision maker $i$, the fusion center employs a permutation mapping $\varrho_i$ such that the fusion rule is equivalent to a monotone $\vartheta_i$ together with a regular likelihood ratio test at the fusion center. Hence, Theorem~\ref{theorem2} and accordingly Theorem~\ref{theorem1} follow.
\end{proof}
\noindent The task of fusion center is to employ an overall permutation mapping $\boldsymbol\varrho=\{\varrho_1,\ldots,\varrho_K\}$ to the received discrete multilevel decisions $u_1,\ldots,u_K$. The mapping described by $\varrho_i$ is well known and can be found in \cite[p. 310]{Tsi}. Notice that fusion center must know which decision corresponds to which decision maker to be able to perform this task.\\
The second possible design of $\boldsymbol{\phi}$ can be achieved through choosing $\varphi_i$ as a trivial function and $\vartheta_i$ as a random function. The following theorem details this claim.
\begin{thm}\label{theorem3}
Let $\varphi_i$ be an identity mapping $Y_i\mapsto X_i$ and let the function ${\vartheta_i:\Omega_i\rightarrow\{0,1\}}$ with the random mapping $\vartheta_i:X_i\mapsto U_i$ results in
\begin{equation}\label{eq18}
\hat{U}_i^0\succeq U_i^0\quad\mbox{and}\quad U_i^1\succeq \hat{U}_i^1\quad \forall i, Q_j^i\in\mathscr{P}_j^i
\end{equation}
which satisfies $\hat{q}_1^i(U_i=0)+\hat{q}_0^i(U_i=1)<1$. Then, all conclusions of Theorem~\ref{theorem1} follow.
\end{thm}
\begin{proof}
It is assumed by \eqref{eq18} that $\vartheta_i$ satisfies stochastic ordering condition imposed on $U_i$. What remains to be shown is that $\hat{l}_i$ is a.e. equal to a non-decreasing function. This condition is true because
\begin{equation*}
\hat{q}_1^i(U_i=0)<1-\hat{q}_0^i(U_i=1)\quad\mbox{and}\quad \hat{q}_0^i(U_i=1)<1-\hat{q}_{1,U_i}(U_i=0),\quad\forall i
\end{equation*}
implies
\begin{equation*}
\hat{q}_0^i(U_i=1)\hat{q}_1^i(U_i=0)<(1-\hat{q}_0^i(U_i=1))(1-\hat{q}_1^i(U_i=0)),\quad\forall i
\end{equation*}
which is
\begin{equation*}
\hat{l}_i(U_i=1)=\frac{1-\hat{q}_1^i(U_i=0)}{\hat{q}_0^i(U_i=1)}>\frac{\hat{q}_1^i(U_i=0)}{1-\hat{q}_0^i(U_i=1)}=\hat{l}_i(U_i=0),\quad\forall i.
\end{equation*}
\end{proof}
\noindent Both Theorem~\ref{theorem2} and Theorem~\ref{theorem3} imply Theorem~\ref{theorem1}. From Theorem~\ref{theorem1} to the inequalities given by \eqref{eq13}, what remains to be shown is that among all possible $\boldsymbol\phi\in\boldsymbol\Delta$, $\hat{\boldsymbol\phi}$ minimizes the overall error probability $P_E$. The problem definition is generic and depending on the choice of uncertainty classes $\mathscr{P}_0^i$ and $\mathscr{P}_1^i$ the decision and fusion rules, $\hat{\boldsymbol{\phi}}$ and $\hat{\gamma}$, may vary.

\section{Examples}\label{examples}
As mentioned in the previous section, the choice of the uncertainty classes may lead to different types of minimax robust test. In this section, three different uncertainty classes are introduced, one of which makes use of Theorem~\ref{theorem2} and the other two Theorem~\ref{theorem3} such that together with $\hat{\boldsymbol{\phi}}$, which minimises $P_{E}$, imply the saddle value inequalities given by \eqref{eq13}.

\subsection{Huber's Extended Uncertainty Classes}
Let us assume that $\mathscr{P}_0^i$ and $\mathscr{P}_1^i$ are given by Huber's extended uncertainty classes, cf. \cite{hube68}, \cite[p. 271]{hube81}, which include various uncertainty classes as special cases such as $\epsilon-$contamination, total variation, Prohorov, Kolmogorov and Levy neighborhoods. For these uncertainty classes the stochastic ordering defined by \eqref{eq14} hold letting $\varphi_i$ to be the robust likelihood ratio functions $\hat{l}_i$ obtained from the related uncertainty classes. Furthermore, if $\vartheta_i$s are monotone non-decreasing functions, or just bijective mappings, see Corollary~\ref{corollary}, Theorem~\ref{theorem2} follows. If additionally $Y_i$ are mutually independent, the optimum mappings $\vartheta_i$ which minimize $P_E$ are known to be the likelihood ratio tests \cite{Tsi}. Hence, Theorem~\ref{theorem2} and the saddle value condition \eqref{eq13} follow. This result was obtained previously by \cite{Poor1} under the assumption that the uncertainty classes satisfy joint stochastic boundedness property.

\subsection{Uncertainty Classes Based on KL-divergence}
The KL-divergence is a smooth distance and hence can be used to design minimax robust tests if the uncertainties are caused by modeling errors or model mismatch, cf. Proposition~\ref{prop1}, \cite{levy09}. The general version of the minimax robust test based on the KL-divergence distance, which is called the (m)-test accepts user defined pair of robustness parameters $(\epsilon_0,\epsilon_1)$ and the pair of nominal distributions $(P_0,P_1)$ and gives a unique pair of least favorable density functions $(\hat{q}_0,\hat{q}_1)$ and a randomized robust decision rule $\hat\phi$ \cite{gul5}. The robustness parameters should be chosen so that the hypotheses do not overlap, i.e. a minimax robust test exists. Existence of a minimax robust test implies $\hat{q}_1^i(U_i=0)+\hat{q}_0^i(U_i=1)<1$ for every decision maker $i$. Moreover, the existence of a saddle value condition stated by \cite{gul5} implies stochastic ordering of $U_i$, i.e. \eqref{eq18}. Hence, by Theorem~\ref{theorem3}, Theorem~\ref{theorem1} follows. Unlike Huber's minimax robust test, for the (m)-test the decision and fusion rules cannot be jointly minimised since $\hat\phi$ is unique and minimises the error probability of every decision maker $P_{E_i}$, not the global error probability $P_E$. Minimizing $P_{E_i}$ for every decision maker does not guarantee that $P_E$ is also minimised. However, there are  special cases, for which $P_E$ is also minimised by $\boldsymbol\hat\phi$. Assume that $\epsilon_0=\epsilon_0^i$, $\epsilon_1=\epsilon_1^i$ and $P_0=P_0^i$, $P_1=P_1^i$. Then, $\boldsymbol\hat\phi$ will be composed of identical decision rules. For identical decision rules, there are also counterexamples showing that no fusion rule $\hat\gamma$ is a minimiser, because identical decision rules are not always optimum \cite{Kantor}. However, for the majority of decision making problems, i.e. for the choice of the probability distributions $P_0$ and $P_1$, identical decision makers are optimum and minimise $P_E$ for some $\hat\gamma$. Similarly, if no assumption is made on the choice of the robustness parameters and the nominal distributions, there are some decision making problems for which $P_E$ is minimised by $\boldsymbol{\hat\phi}$. This result together with Theorem~\ref{theorem1} implies the saddle value condition \eqref{eq13} and thus generalizes \cite{Poor1}, which requires stochastic ordering of random variables $X_i$. Notice that since no other decision rules apart from $\boldsymbol{\hat\phi}$ are able to achieve the saddle value condition defined on $\boldsymbol{U}$, by Theorem~\ref{theorem1} no other decision rules can be minimax robust while minimising $P_E$ either.

\subsection{Uncertainty Classes Based on $\alpha$-divergence}
Similar to the KL-divergence, for the choice of $\alpha-$divergence, $X_i$s are not jointly stochastically bounded, because minimax decision rules are randomised \cite{gul6}. However, a minimax decentralized detection is possible with the same arguments stated in the previous section. The advantage of $\alpha$-divergence over the KL-divergence is that both the distance, namely the parameter $\alpha$, as well as the thresholds of the nominal test $t$ can be chosen arbitrarily for every sensor $i$. This provides flexibility and a more likely scenario that the designed decision rules $\boldsymbol{\phi}$ minimise not only $P_{E_i}$ but also $P_{E}$, hence they also imply the left inequality in \eqref{eq13}. For both schemes, without imposing any additional constraints on the choice of the parameters or the nominal distributions, the right inequality in \eqref{eq13} is always satisfied. Therefore, the power of the test is guaranteed to be above a certain threshold, despite the uncertainty on the sensor network.

\subsection{Composite Uncertainty Classes}
The uncertainty classes for each decision maker can be chosen arbitrarily either from Huber's extended uncertainty classes or from the uncertainty classes formed with respect to the $\alpha-$divergence\footnote{As $\alpha\rightarrow 1$, the $\alpha-$divergence tends to the KL-divergence.}. Based on the information from the previous sections, it can be concluded that the decentralized detection network is minimax robust, if the sensor and the fusion thresholds minimize the overall error probability $P_E$ for the least favorable distributions $\boldsymbol{\hat{Q}_0}$ and $\boldsymbol{\hat{Q}_1}$.

\section{Generalizations}\label{general}
\subsection{Neyman-Pearson Formulation}\label{neymanpearson}
The Neyman-Pearson (NP) version of the same problem can be stated as follows:
\begin{equation}\label{eq19}
\inf_{\boldsymbol\phi\in\boldsymbol\Phi,\gamma}\sup_{\boldsymbol{Q_1}\in\mathscr{P}_1} P_M(\boldsymbol{{Q}_1},{\boldsymbol\phi},{\gamma})\quad \mathrm{s.t.} \,\,\, \sup_{\boldsymbol{Q_0}\in\mathscr{P}_0} P_F(\boldsymbol{Q_0},{\boldsymbol\phi},{\gamma})\leq t.
\end{equation}
If a pair of LFDs $(\boldsymbol{\hat{Q}_0},\boldsymbol{\hat{Q}_1})$ solves the maximisation of the Bayesian version of the minimax optimisation problem \eqref{eq12}, it also solves the maximisation of its NP counterpart \eqref{eq19}, because, the inequalities in \eqref{eq14} imply \eqref{eq19}.\\
For the minimisation, dependently randomised decision and/or fusion rules may need to be employed at sensors, if the distribution of $\hat{l}_i(Y_i)$ has a jump discontinuity under $\mathcal{H}_0$ or $\mathcal{H}_1$, and at the fusion center, cf. \cite{Tsi,Varshney}. While randomisation may be allowed to solve \eqref{eq19} if Huber's uncertainty classes are considered, the same conclusion cannot be made thoroughly when the uncertainty classes are constructed based on the $\alpha-$divergence. In the latter case, dependently randomised decision rules may only be allowed at the fusion center but not at the decision makers, because the robust decision rules are unique and modifying them automatically results in the loss of saddle value inequalities \eqref{eq13}, \cite{gulbook}.

\subsection{Repeated Observations}\label{repeated}
The proposed model includes the case, where one or more decision makers give their decisions based on a block of observations $\mathbf{y}_i=(y_i^1,\ldots,y_i^n)$, which are not necessarily obtained from identically distributed random variables $Y_i^1,\ldots,Y_i^n$. For every decision maker $i$, if the Huber's uncertainty classes are considered, it is known that the multiplication of the robust likelihood ratio functions also satisfies the minimax condition \cite[p. 1756]{hube65}. However, this is not true when the uncertainty classes are constructed for the $\alpha$-divergence distance, cf.~\cite{gul5}. In this case, the minimax tests must be designed over multi-variate distribution functions.

\subsection{Different Network Topologies}\label{diffnet}
Among the network topologies, probably the parallel network topology has received the most attention in the literature \cite{Varshney}. However, depending on the application, decentralized detection networks can be designed considering a number of different topologies, for example a serial topology, a tree topology, or an arbitrary topology \cite{Tsi}. For arbitrary network topologies, it is known that likelihood ratio tests are no longer optimal, in general \cite[p. 331]{Tsi}. Therefore, the results obtained for a parallel network topology cannot be generalized to arbitrary networks in a straightforward manner. Each network structure requires a new and possibly much complicated design. In light of Theorem~\ref{theorem1}, obtaining bounded error probability at the output of the fusion center is easier. Every sensor in the sensor network is required to transmit stochastically ordered decisions to its neighboring sensors and must make sure that the average error probability is less than $1/2$. This guarantees bounded error probability. Minimization of the global error probability $P_E$ can be handled separately.\\
Asymptotically, i.e. when the number of sensors goes to infinity, $P_E$ goes to zero if the network topology is parallel. This is a consequence of Cramer's Theorem~\cite{cramer} for Bernoulli random variables $U_i$. If the network of interest is a tandem network, the error probability is almost surely bounded away from zero if $l_i$ for every sensor $i$ is bounded under each hypothesis $\mathcal{H}_j$ \cite{cover1969, hellman1970}. Remember that Huber's clipped likelihood ratio test bounds the nominal likelihoods, therefore, a minimax robust tandem network can never be asymptotically error free \cite{serial}. On the other hand, the minimax robust test based on the KL-divergence or $\alpha$-divergence does not alter the boundedness properties of $l_i$s, hence, preserves the asymptotic properties of the network.

\section{Conclusions}\label{conclusion}
In this paper, minimax robust decentralized hypothesis testing has been studied for parallel sensor networks. It has been proven that the minimax robust tests designed from the KL-divergence neighborhood do not satisfy the joint stochastic boundedness property. This has motivated an attempt to prove whether minimax robust decentralized detection is possible in this case. The theory has been developed under the assumption that the random variables $Y_i$ corresponding to the observations $y_i$ are independent but not necessarily identical. Additionally, multi-level quantisation at decision makers was also allowed. Three examples of the proposed robust model has been provided. An extension of the proposed model to the Neyman-Pearson test, repeated observations, and different network topologies has been discussed. The proposed model generalizes \cite{Poor1} since stochastic boundedness property is not required at sensors and the sensors decision rules do not have to be monotone in order to achieve minimax robustness. This allows different types of minimax robust tests to be simultaneously employed by the decision makers, not only the clipped likelihood ratio tests.

The open problems arising from this work can be listed as follows:
\begin{itemize}
\item What are the minimax strategies for the sensor networks with arbitrary topologies, for which likelihood ratio test is known not to be optimum?
\item How does the design, i.e. $\varphi_i$ and $\vartheta_i$ should look like when $Y_i$ are not mutually independent in order to guarantee bounded error probability?
\end{itemize}

\bibliographystyle{IEEEtran}
\bibliography{strings2new}
\end{document}